\documentclass[final,3p,times]{elsarticle}

\usepackage{amssymb}

\usepackage{amsmath}
\usepackage{amsthm}
\usepackage{txfonts}
\usepackage{dsfont}
\usepackage{amssymb}
\usepackage{amsfonts}
\usepackage{algorithm}
\usepackage{algorithmic}

 \journal{Elsevier}

\begin{document}

\newtheorem{theorem}{Theorem}[section]
\newtheorem{proposition}{Proposition}[section]
\newtheorem{lemma}{Lemma}[section]
\newtheorem{corollary}{Corollay}[section]
\newtheorem{example}{Example}[section]
\newtheorem{remark}{Remark}[section]

\renewcommand{\algorithmicrequire}{\textbf{Input:}}
\renewcommand{\algorithmicensure}{\textbf{Output:}}
\newcommand{\abs}[1]{\lvert#1\rvert}

\begin{frontmatter}



\title{Parallel computation of real solving bivariate polynomial systems by zero-matching method\tnoteref{t1}}
\tnotetext[t1]{This research was partly supported by China 973
Project NKBRPC-2004CB318003, the Knowledge Innovation Program of the
Chinese Academy of Sciences KJCX2-YW-S02, and the National Natural
Science Foundation of China(Grant NO.10771205).}


\author[UESTC,CASIT,GUCAS]{Xiaolin Qin\corref{cor2}}
\ead{qinxl@casit.ac.cn}
\author[UESTC]{Yong Feng\corref{cor1}}
\ead{yongfeng@casit.ac.cn}
\author[CASIT,GUCAS]{Jingwei Chen}
\author[UESTC,CASIT]{Jingzhong Zhang}

\cortext[cor1]{Corresponding author}
\cortext[cor2]{Principal
corresponding author}

\address[UESTC]{Laboratory of Computer Reasoning and Trustworthy Computation, University of Electronic Science and Technology of China, Chengdu 610054,
China}
\address[CASIT]{Laboratory for Automated Reasoning and Programming, Chengdu Institute of Computer Applications, CAS, Chengdu 610041, China}
\address[GUCAS]{Graduate University of Chinese Academy of Sciences, Beijing 100049, China}

\begin{abstract}
We present a new algorithm for solving the real roots of a bivariate
polynomial system $\Sigma=\{f(x,y),g(x,y)\}$ with a finite number of
solutions by using a zero-matching method. The method is based on a
lower bound for bivariate polynomial system when the system is
non-zero. Moreover, the multiplicities of the roots of $\Sigma=0$
can be obtained by a given neighborhood. From this approach, the
parallelization of the method arises naturally. By using a
multidimensional matching method this principle can be generalized
to the multivariate equation systems.
\end{abstract}


\end{frontmatter}


\section{Introduction}
\label{}
 Considering the following system:
\begin{equation}\label{eq:1}
\Sigma = \{f(x,y),g(x,y)\}, \end{equation} we assume that
$f(x,y),g(x,y) \in \mathbb{Q}[x,y]$, where $\mathbb{Q}$ is the field
of rational numbers. We call the \textbf{zero-dimension} if the
bivariate polynomial system ($\ref{eq:1}$) has a finite number of
solutions.

Real solving bivariate polynomial system in a real field is an
active area of research. It is equivalent to finding the
intersections of $f(x,y)$ and $g(x,y)$ in the real plane. The
problem is closely related to computing the topology of a plane real
algebraic curve and other important operations in non-linear
computational geometry and Computer-Aided Geometric
Design\cite{ACM1984,GN2002,GL2004,EKW2007,KSP2008}. Another field of
application is the quantifier elimination\cite{C1975,J1997}. There
are several algorithms that tackle this problem such as the
Gr$\ddot{o}$ber basis method\cite{L2009,R1999}, the resultant method
\cite{YZH1996}, the characteristic set method \cite{CGY2008}, and
the subdivision method\cite{B2006,MP2008}. However, the procedure of
these techniques is very complicated. In this paper, we propose an
efficient approach to remedy these drawbacks.

In this paper, we propose a zero-matching method to solve the real
roots of an equation system like ($\ref{eq:1}$). The basic idea of
zero-matching method is as follows: First projecting the roots of
$\Sigma$ to the $x$-axis, gives the roots \{$x_{1}, \cdots,
x_{u}$\}, and the $y$-axis, gives the roots \{$y_{1}, \cdots,
y_{v}$\}, respectively. Subsequently, for every root $x_{i}$, and
for every $y_{j}$ is backsubstituted in $f(x,y)$ and $g(x,y)$. To
that end, for some root $x_{i}$ there is the corresponding one or
more roots $y_j$ to be determined satisfying $\Sigma$. The main
contribution of our method is that how to determine the real roots
of $\Sigma=0$ and the multiplicities of the roots. Moreover, our
approach that has given solutions to this situation can be the
design of parallelized algorithms.

In \cite{DET2007}, Diochnos \emph{et al} presented three algorithms
to real solving bivariate systems and analyzed their asymptotic bit
complexities. Among the three algorithms, the difference is the way
they match solutions. The method of specialized Rational Univariate
Representation \textsc{(rur)} based on fast \textsc{gcd}
computations of polynomials with coefficients in an extension field
to achieve efficiency (hence the name \textsc{g\_rur}) has the
lowest complexity and performs best in numerous experiments. The
\textsc{grur} method projects the roots to the $x-$axis and
$y-$axis, for each $x$-coordinate $\alpha$ computes the \textsc{gcd}
$h(\alpha,y)$ of the square-free parts of $f(\alpha,y)$ and
$g(\alpha,y)$, and isolates the roots of $h(\alpha,y)=0$ based on
computations of algebraic numbers and the \textsc{rur} techniques.
Our algorithm only uses resultant computation and real solving for
univariate polynomial equations with rational coefficients.

The hybrid method proposed by Hong \emph{et al}\cite{HSZ2008} that
projects the roots of $\Sigma$ to the $x$-axis and $y$-axis
respectively and uses the improved slope-based Hansen-Sengupta to
determine whether the boxes formed by the projection intervals
contain a root of $\Sigma$. The numerical method only works for
simple roots of $\Sigma$. When the system has multiple roots, the
\textsc{rur} technique is used to isolate the roots. Compared with
this method, our approach also computes two resultants of the same
total degrees. However, our method is a complete one, their
numerical iteration method needs to use the \textsc{rur} technique
to find multiple roots.

In\cite{BBG2005}, Bekker \emph{et al} presented a Combinatorial
Optimization Root Selection method (hence the name \textsc{cors}) to
match the roots of a system of polynomial equations. However, the
method is only suitable for solving a small system of polynomial
equations, and does not work for the multiple roots. Recently, Cheng
\emph{et al}\cite{CGL2008} proposed a local generic position method
to solve the bivariate polynomial equation system. The method can be
used to represent the roots of a bivariate equation system as the
linear combination of the roots of two univariate equations.
Moreover, the multiplicities of the roots of the bivariate
polynomial system are also derived. However, the method is very
complicated to extend to solve the multivariate equation systems.
Our method can solve the larger systems and easily generalize to the
multivariate equation systems.

The rest of this paper is organized as follows. In Section 2, we
give some notations, a lower bound for bivariate polynomial equation
if it is non-zero, and how to determine the root multiplicity. In
Section 3, we propose the algorithm to real solving the bivariate
polynomial system and give a detailed example. In section 4, we
present some comparisons of our algorithm. The final section
concludes this paper.

\section{Notations and main results}
\subsection{Notations}
In what follows \textbf{D} is a ring, \textbf{F} is a commutative
field of characteristic zero and $\overline{\textbf{F}}$ its
algebraic closure. Typically $\textbf{D}=\mathbb{Z}$,
$\textbf{F}=\mathbb{Q}$ and $\overline{\textbf{F}}
=\overline{\mathbb{Q}}$.

In this paper, we consider the zero-dimensional bivariate polynomial
system as follows:
\begin{equation} \label{eq:2}
\left\{ \begin{aligned}
        f(x,y) &= \sum_{0\leq i \leq n}\sum_{0\leq j \leq m}a_{i,j}x^{i}y^{j}=0\\
                  g(x,y)&=\sum_{0\leq i \leq p}\sum_{0\leq j \leq q}b_{i,j}x^{i}y^{j}=0\\
                          \end{aligned} \right. .
                          \end{equation}
Throughout this paper, note that $deg_{x}=\max{(n,p)}$,
$deg_{y}=\max{(m,q)}$, $N=\max{(||f||_{1},||g||_{1})}$, where the
$||f||_{1}$ and $||g||_{1}$ are the
 one norm of the vector $(a_{00},a_{01},\cdots, a_{0m},\cdots, a_{n0},\cdots,a_{nm})$ and
$(b_{00},b_{01},\cdots,b_{0q},\cdots,b_{p0},\cdots,b_{pq})$, so
$||f||_{1}=\Sigma_{i}\Sigma_{j}|a_{ij}|$, and
$||g||_{1}=\Sigma_{i}\Sigma_{j}|b_{ij}|$, respectively. $M =\max
(||t||_{1}, ||T||_{1})$, where the $t(x)$ and $T(y)$ are the no
extraneous factors in resultant polynomial of $\Sigma$. $|\Sigma|$
denotes that the bivariate polynomial system $\Sigma$ has been
assigned values to two variables.

Let $\pi$ be the projection map from the $\Sigma$ to the $x$-axis:
\begin{equation}\label{eq:8}
\pi: \mathbb{R}^{2}\rightarrow\mathbb{R}, \ \ \ \ \mbox{such that} \
\pi(x, y)=x.
\end{equation}
For a zero-dimensional system $\Sigma$ defined in ($\ref{eq:2}$),
let $t(x)\in \mathbb{Q}[x]$ be the resultant of $f(x,y)$ and
$g(x,y)$ with respect to $y$:
\begin{equation}
t(x)=\mbox{Res}_{y}(f(x,y),g(x,y)).
\end{equation}
Since $\Sigma$ is zero-dimensional, we have $t(x)\nequiv 0$. Then
$\pi(\textbf{V}(\Sigma))\subseteq \textbf{V}(t(x))$, where
$\textbf{V}(f_{1},\cdots,f_{m})$ is the set of common real zeros of
$f_{i}=0$. If $t(x)$ is irreducible, then denote the highest degree
by $deg_{t}$. Let the real roots of $t(x)=0$ be
\begin{equation}\label{eq:alpha}
\alpha_{1}<\alpha_{2}<\cdots<\alpha_{u}.
\end{equation}
By using the same method, let $T(y)\in \mathbb{Q}[y]$ be the
resultant of $f(x,y)$ and $g(x,y)$ with respect to $x$:
\begin{equation}
T(y)=\mbox{Res}_{x}(f(x,y),g(x,y)).
\end{equation}
If T(y) is irreducible, then denote the highest degree by $deg_{T}$.
Let the real roots of $T(y)$ be as follows:
\begin{equation}\label{eq:beta}
\beta_{1}<\beta_{2}<\cdots<\beta_{v}.
\end{equation}
We observe that the above projection map may generate extraneous
roots. Fortunately, we can easily discard these extraneous factors
by computing the determinant of the sub-matrix of the coefficient
matrix. Moreover, if the resultant is irreducible, then it is no
extraneous factors. However, when the resultant is reducible, it may
suffer from the extraneous factors. The method of removing
extraneous factors mentioned can be adapted to the resultant for the
bivariate polynomial system \cite{ZF2009}. It is the following
theorem to remove the extraneous roots.
\begin{theorem}\label{thm:remove}
$\Sigma$ is defined as in ($\ref{eq:2}$). If the resultants of \
$\Sigma$ for one variable is reducible, denoted by $tem$, then the
resultant of bivariate polynomial system is the only some
irreducible factors in which the other variable appear.
\end{theorem}
\begin{proof}
The proof can be given similarly to that in Proposition 4.6 of
Chapter 3 of \cite{CLS2005}.
\end{proof}
\subsection{A lower bound for $|\Sigma|$, if \ $\Sigma \neq 0$}
The purpose of this subsection is to prove the following theorem.
\begin{theorem}\label{thm:1}
$\Sigma$ is defined as in ($\ref{eq:2}$). Let $\alpha, \beta$ be two
approximate real algebraic numbers. Denote by the integer
$s=deg_{t}\cdot deg_{T}$, and $N$ as above. If \ $|\Sigma| \neq 0$,
then
\begin{equation}\label{eq:3}
|\Sigma |\geq N^{1-s}M^{-c\cdot s},
\end{equation}
where $c$ is the constant satisfying certain conditions, $|\Sigma|$
is the following two cases:\\
(a) If $f(\alpha, \beta)=0$ or $g(\alpha, \beta)=0$, then $|\Sigma|=
\max\{|f(\alpha, \beta)|, |g(\alpha, \beta)|\}$;\\
(b) If $f(\alpha, \beta)\neq 0$ and $g(\alpha, \beta) \neq 0$, then
$|\Sigma|= \min\{|f(\alpha, \beta)|, |g(\alpha, \beta)|\}$.
\end{theorem}
Before giving the proof of theorem $\ref{thm:1}$, we recall two
lemmas:
\begin{lemma}\label{lem:1}
(\cite {MW1978}, lemma 3) Let $\alpha_{1},\dots,\alpha_{q}$ be
algebraic numbers of exact degree of $d_{1},\dots,d_{q}$
respectively. Define
$D=[\mathbb{Q}(\alpha_{1},\dots,\alpha_{q}):\mathbb{Q}]$. Let $P\in
\mathbb{Z}[x_{1},\dots,x_{q}]$ have degree at most $N_{h}$ in
$x_{h}$($1\leq h\leq q$). If $P(\alpha_{1},\dots,\alpha_{q})\neq 0$,
then
\[
|P(\alpha_{1},\dots,\alpha_{q})|\geq \parallel
P\parallel_{1}^{1-D}\prod_{h=1}^{q}M(\alpha_{h})^{-DN_{h}/d_{h}},
\]
where the $M(\alpha_{h})$ is the Mahler measure of $\alpha_{h}$.
\end{lemma}
\begin{proof}
See the Lemma 4 of \cite{MW1978}.
\end{proof}
\begin{lemma}\label{lem:2}
Let $\alpha$ be an algebraic number. Denote by the $M(\alpha)$ of
the Mahler measure of $\alpha$. If $P$ is a polynomial over
$\mathbb{Z}$, then
  $$M(\alpha)\leq ||P||_{1}.$$
\end{lemma}
\begin{proof}
For any polynomial $P=\sum_{i=0}^{d}p_{i}\in \mathbb{Z}[x]$ of
degree $d$ with the all roots $\sigma^{(1)},\cdots,\sigma^{(d)}$, we
define the $measure$ $M(P)$ by
\begin{eqnarray*}
M(P)=|p_{d}|\Pi_{i=1}^{d}\max{\{1,|\sigma^{(i)}|\}}.
\end{eqnarray*}
The Mahler measure of an algebraic number is defined to be the
Mahler measure of its minimal polynomial over $\mathbb{Q}$. We know
from Landau (\cite{GG1999}, p. 154, Thm. 6. 31) that for each
algebraic number $\alpha$
\begin{eqnarray*}
M(\alpha)\leq ||P||_{2},
\end{eqnarray*}
where $||P||_{2}=(\sum_{i=0}^{d}|p_{i}|^{2})^{1/2}$. It is very easy
to get that $||P||_{2}\leq ||P||_{1}$. This completes the proof of
the lemma.
\end{proof}
Now we turn to give the proof of Theorem $\ref{thm:1}$.
\begin{proof}
From the assumption of the theorem, since $\Sigma$ is defined as in
(\ref{eq:2}). Let the pair ($\alpha$, $\beta$) be corresponding
value to the variable $x$ and $y$ for $\Sigma$ respectively. We have
the following equations:
\begin{subequations}
\begin{eqnarray}
&&f(\alpha,\beta)=
\sum_{0\leq i \leq n}\sum_{0\leq j \leq m}a_{i,j}\alpha^{i}\beta^{j} \label{eq:4a}\\
&&g(\alpha,\beta)=\sum_{0\leq i \leq p}\sum_{0\leq j \leq
q}b_{i,j}\alpha^{i}\beta^{j}. \label{eq:4b}
\end{eqnarray}
\end{subequations}
At first, we consider the lower bound for the equation
($\ref{eq:4a}$). Define $k=[\mathbb{Q}(\alpha, \beta): \mathbb{Q}]$.
Denote by $|f|$=$|f(\alpha, \beta)|$, and $r$, $t$ by the exact
degree of algebraic numbers $\alpha$, $\beta$ respectively. From
Lemma (\ref{lem:1}), if $|f|\neq 0$, then
\begin{eqnarray*}
|f|\geq ||f||^{1-k}_{1}M(\alpha)^{-kn/r}M(\beta)^{-km/t}.
\end{eqnarray*}
We observe that $M(\alpha)$ and $M(\beta)$ derive from $t(x)$ and
$T(y)$ respectively. From Lemma ($\ref{lem:2}$), we can get the
following inequality:
\begin{eqnarray*}
M(\alpha)\leq ||t||_{1}, M(\beta)\leq ||T||_{1}.
\end{eqnarray*}
So we can obtain that
\begin{equation}
|f|\geq ||f||^{1-k}||t||_{1}^{-kn/r}||T||_{1}^{-km/t}.
\end{equation}
By using the same technique as above, we can obtain the lower bound
for the equation ($\ref{eq:4b}$). Denote by $|g|$=$|g(\alpha,
\beta)|$. If $|g|\neq 0$, then
\begin{equation}
|g|\geq ||g||^{1-k}||t||_{1}^{-kn/r}||T||_{1}^{-km/t}.
\end{equation}
Since we have the following two cases:\\
(a) If $f(\alpha, \beta)=0$ or $g(\alpha, \beta)=0$, then $|\Sigma|=
\max\{|f(\alpha, \beta)|, |g(\alpha, \beta)|\}$;\\
(b) If $f(\alpha, \beta)\neq 0$ and $g(\alpha, \beta) \neq 0$, then
$|\Sigma|= \min\{|f(\alpha, \beta)|, |g(\alpha,
\beta)|\}$.\\
Hence we are able to obtain the lower bound for the bivariate
polynomial system. From the above assumption, we can get the
following parameters:
\begin{subequations}
\begin{eqnarray}
&&k=[\mathbb{Q}(\alpha, \beta): \mathbb{Q}]\leq deg\{t(x)\}\cdot deg\{T(y)\}=deg_{t}\cdot deg_{T}, \label{eq:7a}\\
&&N=\max\{||f||_{1}, ||g||_{1} \}, M=\max\{||t||_{1}, ||T||_{1}\},
r=deg_{t}, t=deg_{T}.\label{eq:7b}
\end{eqnarray}
\end{subequations}
Combined with the equation (\ref{eq:7a}) and (\ref{eq:7b}), it is
obvious that $s=k$ and the constant
$c=\frac{deg_{t}}{r}+\frac{deg_{T}}{t}+1$. Finally, note that the
constant $c$ satisfies both cases. This proves the theorem.
\end{proof}
As the corollary of Theorem $\ref{thm:1}$, we have
\begin{corollary}\label{cor:1}
Under the same condition of Theorem $\ref{thm:1}$, if \ $|\Sigma |<
N^{1-s}M^{-c\cdot s}$, then $|\Sigma|=0$. We say that $\alpha$ is
\textbf{associated with} $\beta$ for the real root of \ $\Sigma$ .
Denote by the $\varepsilon=N^{1-s}M^{-c\cdot s}$ for the rest of
this paper.
\end{corollary}
\begin{proof}
The proof is very easy by contradiction.
\end{proof}

\subsection{Root multiplicity}
The results of this subsection can be provided for the root
multiplicity of $\Sigma$. We follow the approach and terminology of
\cite{CLS2005} and \cite{EIT2005}.

Let $\mathcal{C}_{f}$, $\mathcal{C}_{g}$ be $f$, $g$ corresponding
affine algebraic plane curves, defined by the equations $\Sigma$.
Let $I =<f, g>$ be the ideal that they generate in $\mathbf{F}[x,
y]$, and so the associated quotient ring is $\mathcal{A} =
\overline{\mathbf{F}}[x, y]/I$. Let the distinct intersection
points, which are the distinct roots of ($\Sigma$), be
$\mathcal{C}_{f} \cap \mathcal{C}_{g} \subset \{S_{ij}=(\alpha_{i},
\beta_{j})\}_{1\leq i\leq u,1\leq j\leq v}$.

The multiplicity of a point $S_{ij}$ is
\begin{eqnarray*}
mult(S_{ij}: \mathcal{C}_{f}
 \cap \mathcal{C}_{g}) =
dim_{\overline{\mathbf{F}}}\mathcal{A}_{S_{ij}}< \infty,
\end{eqnarray*}
where $\mathcal{A}_{S_{ij}}$ is the local ring obtained by
localizing $\mathcal{A}$ at the maximal ideal $I =<x-\alpha_i,
y-\beta_j>$.

If $\mathcal{A}_{S_{ij}}$ is a finite dimensional vector space, then
$S_{ij}=(\alpha_i, \beta_j)$ is an isolated zero of $I$ and its
multiplicity is called the intersection number of the two curves.
The finite $\mathcal{A}$ can be decomposed as a direct sum
$\mathcal{A}= \mathcal{A}_{S_{11}}\bigoplus
\mathcal{A}_{S_{12}}\bigoplus \cdots \\ \bigoplus
\mathcal{A}_{S_{uv}}$ and thus
$dim_{\overline{\mathbf{F}}}\mathcal{A}=\sum_{i=1}^{uv}mult(S_{ij}:
\mathcal{C}_{f} \cap \mathcal{C}_{g})$.

\begin{proposition}\label{pro:1}(\cite{EIT2005}, Proposition 1) Let $f$, $g \in
\mathbf{F}[x,y]$ be two coprime curves, and let $p \in
\overline{\mathbf{F}}^2$ be a point. Then
\begin{eqnarray*}
mult(p: fg)\geq mult(p: f)mult(p: g),
\end{eqnarray*}
where equality holds if and only if $\mathcal{C}_f$ and
$\mathcal{C}_g$ have no common tangents at $p$.
\end{proposition}

\begin{proposition}\label{pro:2}
Let us obtain the real roots of \ $\Sigma=0$ in ($\ref{eq:alpha}$)
and ($\ref{eq:beta}$). If the two matching pairs $(\alpha_{i},
\beta_{j})$ and $(\alpha_{i+1}, \beta_{j+1}) (for \ 1\leq i\leq u,
1\leq j\leq v)$ are satisfying $\Sigma=0$,
$|\alpha_{i}-\alpha_{i+1}|<\varepsilon$ and
$|\beta_{j}-\beta_{j+1}|<\varepsilon$, then the $(\alpha_{i},
\beta_{j})$ is multiple root of \ $\Sigma=0$.
\end{proposition}
\begin{proof}
From Theorem \ref{thm:1} and Corollary \ref{cor:1}, it is obvious
that $\Sigma =0$ if and only if
\begin{eqnarray*}
|\Sigma|< \varepsilon.
\end{eqnarray*}
Therefore, the error controlling is less than $\varepsilon$ in
numerical computation. Under the assumption of the proposition, we
get $|\alpha_i-\alpha_{i+1}|< \varepsilon $ and
$|\beta_j-\beta_{j+1}|<\varepsilon$. So we are able to obtain that
$|\alpha_i-\alpha_{i+1}| =0$ and $|\beta_j-\beta_{j+1}| =0$ int the
truncated error. This proves the proposition.
\end{proof}

From Corollary $\ref{cor:1}$, the two-tuple $(\alpha, \beta)$ is the
real root of $\Sigma=0$. This method is called a
\textbf{zero-matching method}. The technique is a posteriori method
to match the solutions for the bivariate system. It can be
generalized easily to real solving the multivariate polynomial
systems.
\section{Derivation of the Algorithm}
The aim of this section is to describe an algorithm for real solving
bivariate polynomial equations by using zero-matching method. We
first find the parameters $N$, $c$ and $s$, then obtain the no
extraneous factors $t(x)$ and $T(y)$ with the resultant elimination
methods, and real solving two univariate polynomials, and finally
match the real roots for the systems.
\subsection{Description of algorithm}
Algorithm 1 is to discard the extraneous factors from the resultant
method, algorithm 2 is to obtain the solutions of bivariate
polynomial systems.

\begin{algorithm}
 \caption{$\mathbf{NoExtrRes}(\Sigma$, $var$)}
\begin{algorithmic}[1]

  \label{alg:1}
  \REQUIRE \{$f(x, y), g(x, y)$\}, $var$ is one variable.
  \ENSURE No extraneous factors resultant of $\Sigma$.
\STATE $tem \leftarrow Res_{var}\{f(x, y), g(x, y)\}$;
\IF {$tem$ is
irreducible}
\RETURN $tem$ ;
\ELSE
\STATE $tem \leftarrow Res \cdot
extraneousfacotrs$;
\RETURN $Res$.
 \ENDIF

\end{algorithmic}
\end{algorithm}

Now we can give the algorithm \ref{alg:2} to compute the real roots
for $\Sigma=0$.

\begin{algorithm}
  \caption{$\mathbf{\textsc{zmm}}(\Sigma$)}
  \begin{algorithmic}[1]
  \label{alg:2}
 \REQUIRE  $\Sigma= \{f(x, y), g(x, y)\}$ is a zero-dimensional bivariate polynomial
 system.
 \ENSURE  A set for the real roots of $\Sigma= 0$.
\STATE Project on the $x$-axis such that $t(x)= Res_y( f(x, y), g(x,
y))$; \STATE Project on the $y$-axis such that $T(y) = Res_x( f(x,
y), g(x, y))$; \STATE Discard the extraneous factors from $t(x)$ and
$T(y)$ by using Algorithm \ref{alg:1}; \STATE Find the parameters
$N$ and $s$, and Compute $c$ according to the Theorem \ref{thm:1};
\STATE Obtain the lower bound $\varepsilon$ by Corollary
\ref{cor:1}; \STATE Solve the real roots of the resultant $t(x)$ for
the set $\mathbf{S_{x}}= \{\alpha_1, \alpha_2,\cdots, \alpha_u\}$;
\STATE Solve the real roots of the resultant $T(y)$ for the set
$\mathbf{S_{y}}= \{\beta_1, \beta_2, \cdots, \beta_v\}$; \STATE
Match the real root pair to get the solving set $\mathbf{S} =
\{(\alpha_i, \beta_j), 1 \leq i \leq u, 1 \leq j \leq v\}$ by
Corollary \ref{cor:1};
 \STATE Check the root multiplicity of the set
$\mathbf{S}$ by Proposition \ref{pro:2}.
\end{algorithmic}
\end{algorithm}

The parallelization of the algorithm that we have just described can
be easily done because it performs the same computations on
different steps of data without the necessity of communication
between the processors. Observe that the Step 1 and Step 2, Step 6
and Step 7 of the algorithm can be easily paralleled, respectively.

Now we get a theorem about the computational complexity of the whole
algorithm.

\begin{theorem}
 Algorithm \ref{alg:2} works correctly as specified and its complexity
includes as follows: \\
(a) $O(d\tau+ dlgd)$ for computation of real solving univariate
polynomial, where $d$ is the degree of corresponding polynomial,
$\tau = 1 + \max_{i\leq d} lg|a_i|$ and $a_i$ is the
coefficients.\\
 (b) $O(uv)$ for matching the solutions of bivariate
polynomial system.
\end{theorem}
\begin{proof}
 Correctness of the algorithm follows from theorem \ref{thm:1}.\\
(a) The number of arithmetic operations required to isolate all real
roots is the number of real root isolation of univariate polynomial
by using subdivision-based Descartes' rule of sign. Using exactly
the same arguments we know that they perform the same number of
steps, that is $O(d\tau + dlgd)$. \\
(b) As indicated before, the problem
of matching the real roots of polynomial system mainly relies on the
scale of solutions of every variable, respectively.
\end{proof}

\subsection{A small example in detail}
\newcounter{num3}
\begin{example} We propose a simple example $f(x,y)=x^{2}-y^{2}-3$ and
$g(x,y)=3x^{2}-2y^{3}-1$ to illustrate our algorithms.
\begin{list}{Step \arabic{num3}:}{\usecounter{num3}\setlength{\rightmargin}{\leftmargin}}
\item $t(x)=4*x^6-45*x^4+114*x^2-109$;
\item $T(y)=(-2*y^3+8+3*y^2)^2$;
\item Discard the extraneous factors $T(y)=-3*y^2-8+2*y^3$;
\item Obtain the parameters $N=5$, $c=2$, $s=4$;
\item Obtain the lower bound $\varepsilon=.1280\times10^{-4}$;
\item Solve the real roots of the resultant $t(x)$ for the set
$\mathbf{S_{x}}=\{-2.858288520, 2.858288520\}$;
\item Solve the real roots of the resultant $T(y)$ for the set
$\mathbf{S_{y}}=\{2.273722337\}$;
\item Combine the the pairs from $\mathbf{S_{x}}$ and $\mathbf{S_{y}}$ respectively,
Substitute the pairs into $\Sigma$ for variables $x$ and $y$,
determine whether less than the lower bound $\varepsilon$, finally
we find that the pairs $\mathbf{S}=\{\{x = -2.858288520, y =
2.273722337\}$, $\{x = 2.858288520, y = 2.273722337\}\}$ are the
solutions for $\Sigma$;
\item The multiplicity of the root of the system is one.
\end{list}
\end{example}

\subsection{Generalization and applications}
As for the generalization of the algorithm to real solving the
multivariate equation systems case, we have to say that the
situation is completely analogous to the bivariate case. However,
its key technique is to transform the multivariate polynomial
equations to the corresponding univariate polynomial equations. We
can consider the Dixon Resultant Method to break this problem
\cite{CZG2002}. However, we observe that how to improve the
projection algorithm in resultant methods is the significant
challenge.

Moreover, our algorithm is applicable for rapidly computing the
minimum distance between two objects collision detection
\cite{YFQ2009}. This also enables us to improve the complexity of
computing the topology of a real plane algebraic curve
\cite{DET2007}.

\section{Some comparisons}
We have implemented the above algorithms as a software package
\textsc{zmm} in \textit{Maple} 12. For problems of small size like
the example of Section 3, any method can obtain the solutions in
little time. But when the size of the problems is not small the
differences appear clearly. Extensive experiments with this package
show that this approach is efficient and stable, especially for
larger and more complex bivariate polynomial systems.

We compare our method with \textsc{lgp} \cite{CGL2008}, Isolate
\cite{R1999}, \textsc{discoverer} \cite{XY2002}, and \textsc{grur}
\cite{DET2007}. \textsc{lgp} is a software package for root
isolation of bivariate polynomial systems with local generic
position method. Isolate is a tool to solve general equation systems
based on the Realsolving \textsc{c} library by Rouillier.
\textsc{discoverer} is a tool for solving semi-algebraic systems.
\textsc{grur} is a tool to solve bivariate equation systems. The
following examples run in the same platform of\textit{ Maple} 12
under Windows and \textsc{amd} Athlon(tm) 2.70 \textsc{ghz}, 2.00
\textsc{gb} of main memory. We did three sets of experiments. The
precision in these experiments is set to be high. In three tables,
where '?' represents that the computation is not finished.

In Table 1 the results are given both $f$ and $g$ are randomly
generated dense polynomials with the same degree and with integer
coefficients between $-20$ and $20$. The command of
\textit{Maple} is as follows: \\
$randpoly([x, y], coeffs = rand(-20..20), dense, degree = 10)$.

\begin{table}[H]\label{tab:1}
\begin{center}
\caption{time for computing dense bivariate polynomials with no
multiple roots}
\begin{tabular}{|c||c|c|c||c|c|c|c|c|}
 \hline \rotatebox{90}{system}& \multicolumn{2}{c|}{deg} & \rotatebox{90}{solutions}&\multicolumn{5}{c|} {\raisebox{2.3ex}[0pt]{Average
 Time(sec)} } \\
\cline{2-3}\cline{5-9}  & $f$ & $g$ & & \textsc{zmm }&\textsc{lgp}&Isolate &\textsc{discoverer} &\textsc{grur}\\
\hline\hline
 S1& 4 &7 &2& 0.031 &0.031 &0.047 &0.313 &2.734 \\ \hline
  S2& 6 &8& 6 &0.415&
1.328 &0.500 &1.828 &247.203 \\ \hline S3& 7 &8 &6& 1.204 &2.734&
1.500 &7.047 &382.640 \\ \hline
 S4 &8& 9 &6 &4.211 &8.906& 4.672 &20.437&
2714.438 \\ \hline S5& 9& 10& 2& 4.070 &8.485& 4.687 &89.235&
1645.312\\ \hline S6 &10 &7 &6& 1.805& 3.860 &2.109& 22.250&
978.421 \\ \hline S7 &10 &11 &4 &21.078& 43.734 &22.828 &?& ? \\
\hline S8 &12& 11 &2 &26.945& 54.969& 29.094& ? &? \\ \hline S9 &12&
13 &4 &118.266 &241.734& 123.469 &? &? \\ \hline  S10 &13 &11& 1&
15.446 &31.485& 17.796 &? &?  \\ \hline S11& 14 &10& 8 &63.914&
200.828 &68.594 &? &? \\ \hline
\end{tabular}
\end{center}
\end{table}

In Table 2 the results are given both $f$ and $g$ are randomly
generated sparse polynomials in the same degree, with sparsity
$default$, and with integer coefficients between $-20$ and
$20$. The command of \textit{Maple} is as follows: \\
$randpoly([x, y], coeffs = rand(-20..20), sparse, degree = 10)$.

\begin{table}[H]\label{tab:2}
\begin{center}
\caption{time for computing sparse bivariate polynomials with no
multiple roots}
\begin{tabular}{|c||c|c|c||c|c|c|c|c|}
 \hline \rotatebox{90}{system}& \multicolumn{2}{c|}{deg} & \rotatebox{90}{solutions}&\multicolumn{5}{c|} {\raisebox{2.3ex}[0pt]{Average
 Time(sec)} } \\
\cline{2-3}\cline{5-9}  & $f$ & $g$ & & \textsc{zmm }&\textsc{lgp}&Isolate &\textsc{discoverer} &\textsc{grur}\\
\hline\hline S1& 5& 6& 1& 0.015& 0.032 &0.015 &0.141&1.032
\\ \hline S2 & 6 &7 &3 &0.040 &0.062& 0.047& 0.188& 5.375 \\ \hline S3
&7& 5 &3& 0.024& 0.047& 0.047& 0.265 &2.688 \\ \hline S4 &8 &6& 5
&0.031 &0.031 &0.047& 0.094 &1.031 \\ \hline S5 &9& 8 &2 &0.047
&0.172 &0.078& 1.828 & 51.000 \\ \hline S6 &10& 11 &3 & 0.063& 0.297
&0.125& 0.656& 11.110 \\ \hline S7 & 11& 9& 2 &0.164& 0.609 &0.375&
3.938&877.875
\\ \hline S8 &12 &13 &2& 1.141& 2.593&1.453& 6.703& 1607.719  \\ \hline S9 &13&
11& 4& 2.508& 5.344 &2.969& ? &? \\ \hline S10&15 &17 &1 &0.532&
1.234& 1.266 &?& ?  \\ \hline S11& 20 &17& 4& 18.180& 39.688&
20.235& ? &?\\ \hline

\end{tabular}
\end{center}
\end{table}
In Table 3 the results are given is done with polynomial systems
with multiple roots. We randomly generate a polynomial $h(x, y, z)$
and take $f(x, y) = Res_z(h, h_z), g(x, y) = f_y(x, y)$. Since $f
(x,y)$ is the projection of a space curve to the $xy$-plane, it most
probably has singular points and $f = g = 0$ is an equation system
with multiple roots. The command of
\textit{Maple} is as follows:\\
$h := randpoly([x, y, z], coeffs = rand(-5..5), degree = 5); f :=
resultant(h, diff(h, z), z); g := diff( f, y).$

\begin{table}[H]\label{tab:3}
\begin{center}
\caption{time for computing bivariate polynomials with multiple
roots}
\begin{tabular}{|c||c|c|c||c|c|c|c|c|}
 \hline \rotatebox{90}{system}& \multicolumn{2}{c|}{deg} & \rotatebox{90}{solutions}&\multicolumn{5}{c|} {\raisebox{2.3ex}[0pt]{Average
 Time(sec)} } \\
\cline{2-3}\cline{5-9}  & $f$ & $g$ & & \textsc{zmm }&\textsc{lgp}&Isolate &\textsc{discoverer} &\textsc{grur}\\
\hline\hline S1 &3 &2 &2& 0.016& 0.016& 0.016 &0.016& 0.062 \\
\hline S2 & 4 &3 &2& 0.& 0.032 &0.031& 0.016& 0.094 \\
\hline S3 &4 &6 &7& 0.024 &0.016& 0.047& 0.109& 1.109 \\
\hline  S4& 5 &4& 3& 0. &0.016& 0.& 0.016& 0.109 \\
\hline S5 &6 &5& 2 &0.015& 0. &0. &0.016& 0.063 \\
\hline S6& 9 &8 &2& 0.016& 0.046& 0.032& 0.015& 0.063 \\
\hline S7 &12 &11& 3 &0.109& 0.234& 0.187& 0.063& 0.094 \\
\hline S8& 13& 12& 7& 2.875& 137.641 &3.141& 1.328 &207.094 \\
\hline S9 &14 &13& 4& 0.860 &2.891& 0.953 &0.141 &0.3110 \\
\hline S10 &19& 18& 1 &0.672& 1.547& 0.797& 22.156 &1520.812 \\
\hline S11 &16 &15 &5 &7.945 &27.047& 9.000 & ?& ? \\
\hline
\end{tabular}
\end{center}
\end{table}
From the Table 1, 2 and 3, we have the following observations.

In the first two cases, the equations are randomly generated and
hence may have no multiple roots. For systems without multiple
roots, \textsc{zmm }is the fastest method, which is significantly
faster than \textsc{lgp} and Isolate. Both \textsc{zmm }and\textsc{
lgp }compute two resultants and isolate their real roots.
\textsc{lgp} is slow, because the polynomials obtained by the shear
map are usually dense and with large coefficients \cite{CGL2008}.
discoverer and \textsc{grur} generally work for equation systems
with degrees not higher than ten within reasonable time.

For systems with multiple roots, in the sparse and low degree cases,
all methods are fast. Note that our method is quite stable for
equation systems with and without multiple roots. \textsc{lgp} and
Isolate are also quite stable, but slower than \textsc{zmm} for
bivariate equation systems.

We also observe that all methods spend more time with sparse and
dense polynomials than polynomials with multiple roots in the same
high degree. This phenomenon needs further exploration.

\begin{remark}
Of course, we should mention that \textsc{discoverer} and Isolate
can be used to solve general polynomial equations and even
inequalities. Here our comparison is limited to the bivariate case.
In further work, we would like to consider solving multivariate
polynomial equations.
\end{remark}
\begin{remark}
  As is well known, the parallel algorithm is well suited for the implementation on parallel computers that
allows the increase of the calculation speed. If our algorithm have
been fully parallelized by using a large enough number of processors
for each case, the real solutions of all the examples will have been
computed in a couple of seconds.
\end{remark}
\section{Conclusion}
In this paper, we propose a zero-matching method to real solving
bivariate polynomial equation systems. The basic idea of this method
is to find the lower bound for bivariate polynomial system when the
system is non-zero. Moreover, we provide an algorithm for discarding
extraneous factors with resultant and show how to construct a
parallelized algorithm for real solving the bivariate polynomial
system. An efficient method for multiplicities of the roots is also
derived. The complexity of our method has increased steadily with
the growth of bivariate polynomial system. Extensive experiments
show that our approach is efficient and stable. The result of this
paper can be extended to real solving of bivariate polynomial
equations with more than two polynomials by using the resultant
method. Furthermore, our method can be generalized easily to
multivariate polynomial systems.

\section{Acknowledgement}
We would like to thank Prof. Xiaoshan Gao and Dr. Jinsan Cheng for
providing the Maple code of their method, available at:
http://www.mmrc.iss.ac.cn/$\sim$xgao/software.html.

The first author is also grateful to Dr. Shizhong Zhao for his
valuable discussions about discarding the extraneous factors in
resultant.

\end{document}